\documentclass[11pt]{article}

\usepackage{hyperref}
\hypersetup{
colorlinks=true,
pdfnewwindow=true,
citecolor=blue}

\newcommand{\poly}{{\rm poly}}
\usepackage{fullpage}
\usepackage{times}
\usepackage{amsmath,amsfonts,amssymb,amsthm}
\usepackage{enumerate}
\usepackage{xcolor}
\usepackage{tikz}
\usetikzlibrary{decorations.pathreplacing}

\newtheorem{definition}{Definition}
\newtheorem{remark}{Remark}
\newtheorem{theorem}{Theorem}
\newtheorem{corollary}{Corollary}
\newtheorem{lemma}{Lemma}

\newtheorem{proposition}{Proposition}
\newtheoremstyle{restate}{}{}{\itshape}{}{\bfseries}{~(restated).}{.5em}{\thmnote{#3}}
\theoremstyle{restate}

\newcommand{\E}{\mathop{\mathbb{E}}}

\newcommand{\bOne}{\mathbf{1}}
\newcommand{\bE}{\mathbf{e}}

\newcommand{\cB}{\mathcal{B}}

\newcommand{\X}{\overline{X}}
\newcommand{\Huff}{\mathsf{Huff}}

\newcommand{\Dic}{\mathsf{Dictionary}_{\Sigma,n}}
\newcommand{\LDSC}{\mathsf{LDSC}}

\newcommand{\Pat}{P\v{a}tra\c{s}cu }

\newcommand{\probarg}[2]{$\mathsf{#1}_{#2}$}

\newcommand{\eps}{\varepsilon}

\newcommand{\cP}{\mathcal{P}}

\newcommand{\cN}{\mathcal{N}}

\newcommand{\cQ}{\mathcal{Q}}
\newcommand{\cT}{\mathcal{T}}

\newcommand{\DPT}{\texttt{DPT}}

\begin{document}

\title{How to Store a Random Walk} 

\author{Emanuele Viola\thanks{Northeastern University.
Email: \texttt{viola@ccs.neu.edu}. Supported by NSF CCF award 1813930.} 
\and 
{Omri Weinstein\thanks{Columbia University. Email: \texttt{omri@cs.columbia.edu}. 
Research supported by NSF CAREER award CCF-1844887.}}
\and 
{Huacheng Yu\thanks{Harvard University. Email: \texttt{yuhch123@gmail.com}. Supported in part by ONR grant N00014-15-1-2388, a Simons Investigator Award and NSF Award CCF 1715187. This work was initiated during a visit of the authors to the Simons Institute.}}
}
\date{}

\setcounter{page}{0}
\maketitle
\thispagestyle{empty}

\maketitle

\abstract{

Motivated by storage applications, we study the following data structure problem:  
An encoder wishes to store a collection of jointly-distributed files 
$\X:=(X_1,X_2,\ldots, X_n) \sim \mu$ which are \emph{correlated} 
($H_\mu(\overline{X}) \ll \sum_i H_\mu(X_i)$), using as little (expected) memory 
as possible, such that each individual file $X_i$ can be recovered quickly with 
few (ideally constant) memory accesses. 

In the case of independent random files, a dramatic result by \Pat (FOCS'08) and 
subsequently by Dodis, \Pat and Thorup (STOC'10) shows that it is possible to store $\overline{X}$
using just a \emph{constant} number of extra bits beyond the information-theoretic minimum space, 
while at the same time decoding each $X_i$ in constant time.  However, in the (realistic) case 
where the files are correlated, much weaker results are known, requiring at least $\Omega(n/\poly\lg n)$ 
extra bits for constant decoding time, even for ``simple" joint distributions $\mu$. 

We focus on the natural case of compressing \emph{Markov chains}, i.e., storing a length-$n$ random 
walk on any (possibly directed) graph $G$. Denoting by $\kappa(G,n)$ the number of  
length-$n$ walks on $G$, we show that there is a succinct data structure storing a random walk 
using $\lg_2 \kappa(G,n) + O(\lg n)$ bits of space, 
such that any vertex along the walk can be decoded in $O(1)$ time on a word-RAM. 
If the graph is strongly connected (e.g., undirected), the space can be improved to only 
$\lg_2 \kappa(G,n) + 5$ extra bits.  For the harder task of matching the \emph{point-wise} optimal 
space of the walk, i.e., the empirical entropy $\sum_{i=1}^{n-1} \lg (deg(v_i))$, we present a data structure 
with $O(1)$ extra bits at the price of $O(\lg n)$ decoding time, and show that 
any improvement on this would lead to an improved solution on the long-standing 
Dictionary problem. All of our data structures  support the \emph{online} version of 
the problem with constant update and query time. 

}

\newpage

\section{Introduction} 

Consider the following 
information-retrieval problem: an encoder (say, a Dropbox server) 
receives a large collection of jointly-distributed files $\overline{X}:=(X_1,X_2,\ldots, X_n) \sim \mu$ which are 
highly \emph{correlated}, i.e.,\footnote{$H_\mu(X)$ denotes the Shannon entropy of $X\sim \mu$}   
\begin{equation} \label{eq_correlation}
H_\mu(\overline{X}) \ll \sum_i H_\mu(X_i),  
\end{equation} 
and needs to preprocess $\overline{X}$ into as little (expected) memory as possible, such that 
each individual file $X_i$ can be 
retrieved quickly using few (ideally constant) memory accesses. 
This data structure problem 
has two naiive solutions: The first one is to compress the entire collection using Huffman 
(or Arithmetic) coding, achieving essentially optimal space $s = H_\mu(\overline{X}) +1$ bits  
(in expectation), but such entropy-coding schemes generally require decompressing the \emph{entire} 
file even if only a single file $X_i$ needs to be retrieved. 
The other extreme is to compress each file separately, achieving (highly) suboptimal space 
$s\sim \sum_i H_\mu(X_i) + n \gg H_\mu(\overline{X})$ at the upshot of constant decoding time $t=O(1)$.
Is it possible to get the best of both worlds (time and space), i.e., a \emph{locally-decodable} data compression scheme?  
\footnote{Of course, 
one could consider a combination of the two solutions by dividing the files into blocks of size 
$\sim t$, and compressing each one optimally, but this method can in general be arbitrarily lossy in 
entropy (space) when the blocks are correlated.}

\vspace{0.05in}

This problem is motivated by large-scale storage applications. 
The proliferation in digital data being uploaded and analyzed on 
remote servers is posing a real challenge of scalability in modern storage systems.  
This challenge is incurred, in part, by the redundancy of maintaining very large yet low-entropy datasets.    
At the same time, in many modern storage applications such as genome sequencing and analysis, 
real-time financial trading, image processing etc., databases are no longer merely serving archival purposes --  
data needs to be continually accessed and manipulated for training, prediction and real-time statistical 
decision making \cite{TBW18, HBBCDDF18, correlationtimesseries, THOW16, Autocorr2000}.  
This inherent tension between \emph{compression and search}, i.e.,  
the need to perform local computations and search over the compressed data 
itself without first decompressing the dataset, 
has motivated the design of compressions schemes which provide \emph{random access} 
to individual datapoints, at a small compromise of the compression rate, giving rise to the notion of 
\emph{locally-decodable source coding} (LDSC) \cite{Pat08, DLRR13, MHMP15, Mazumdar_LDSC, 
TBW18}. 
Local decodability is also a crucial aspect in distributed file systems, where the energy cost of 
random-accessing a memory address is typically much higher than that of sending its actual content,  
especially in SSD hardware \cite{Agrawal_SSD08}. 

There is a long line of work successfully addressing the LDSC problem 
in \emph{i.i.d} or nearly-i.i.d settings (i.e., when \eqref{eq_correlation} holds with equality or  
up to $\pm o(n)$ additive factor  \cite{Pagh02, Buhrman02, GolynskiRR08, Pat08, DPT10, Mazumdar_LDSC, 
MHMP15, TBW18, BN13}), where \emph{succinct} data structures are possible. 
In contrast, the \emph{correlated} case \eqref{eq_correlation} is much less 
understood (more on this in Sections \ref{sec_related_work} and \ref{subsec_completeness_LDSC}). 
Clearly, the aforementioned 
tradeoff between compression and search is all the more challenging (yet appealing) when the dataset 
exhibits many similarities/correlations (i.e., $H_\mu(\overline{X}) \ll \sum_i H_\mu(X_i)$), and this is indeed 
the reality of a large portion of digital media \cite{WSYZW16}. 
Once again, joint-compression of the entire dataset is information-theoretically superior 
in terms of minimizing storage space, but at the same time, global compression, by nature, extinguishes 
any ``locality" in the original data, making it useless for random-access applications.   

\vspace{0.05in}

Unfortunately, a  simple observation (see Proposition \ref{prop_completeness_ldsc}) 
shows that the general LDSC problem in \eqref{eq_correlation} is a ``complete" static data structure 
problem (in the cell-probe model), in the sense that any data structure problem $\cP$ (with an 
\emph{arbitrary} set of $n$ database 
queries) can be ``embedded" as an LDSC problem on \emph{some} joint distribution $\mu = \mu(\cP)$. 
This observation implies that locally-decodable data compression is generally impossible, 
namely, for \emph{most} low-entropy distributions $\mu$ on $n$ files ($H_\mu := H_\mu(\overline{X}) \ll n$), 
any data structure  requires either near-trivial storage space $s \gtrsim n^{1-o(1)}$ or 
decoding time $t\geq (H_\mu)^{1- o(1)}$ for each file (An \emph{explicit} hard example is the family 
of $(H_\mu)$-wise independent distributions, for which our reduction implies $t\geq \Omega(\lg H_\mu)$ 
retrieval time unless super-linear $s= \omega(H_\mu)$ storage space is used).  
These unconditional impossibility results 
naturally raise the question: 
\begin{quote}
Which statistical properties of low-entropy distributions $\mu$ 
facilitate an efficient $\mathsf{LDSC}$ 
scheme with space $O(H_\mu(\overline{X}))$, in the word-RAM model? 
\end{quote}  

Perhaps the most basic family of low-entropy joint distributions is that of 
\emph{Markov chains}, which models ``time-decaying" correlations in discrete-time sequences. 
Ergodicity is a standard statistical model in 
many realistic datasets (e.g., financial time-series, DNA sequences and weather 
forecasting to mention a few \cite{THOW16, Autocorr2000,Gales07}). 
We study the LDSC question on Markov chains, captured by the problem of storing 
a length-$n$ \emph{random walk} on (directed or undirected) graphs of finite size. 
Our data strucutres greatly improve on previously known time-space tradeoffs for Markov chains 
obtained via ``universal compression" (e.g. \cite{DLRR13, TBW18}, see elaboration in 
Section \ref{sec_related_work}), and also efficiently support the dynamic version of the problem 
where the walk evolves step-by-step in an online fashion.  
Throughout the paper, all logarithms are base $2$ unless otherwise stated.  
\subsection{Main Results} 

To build intuition for the problem, let $G$ be an undirected $d$-regular graph of finite 
(but arbitrarily large constant) size, 
and consider a random walk $W=(V_0,V_1,V_2,\ldots, V_n)$ on $G$ starting from a random 
vertex (for simplicity), where $V_i \in [|G|]$ denotes the $i$th vertex $v_i$ visited in the walk. 
Excluding the first vertex $V_0$ for convenience of notation, the joint entropy 
of $W$ is clearly $H(W) = \lg|G|+n\lg d$, whereas the sum-of-marginal entropies is 
\begin{align}\label{eq_RW_entropy}
\sum_{i=0}^{n} H(V_i) = \; (n+1)\lg |G| \; \gg \; n \lg d  
\end{align}
since each $V_i$ is marginally uniform (as $G$ was assumed to be regular and 
we started from a uniform vertex).  This simple instance already captures an interesting case of the 
LDSC problem: The information-theoretic minimum space for storing the walk ($\lg|G| + n\lg d$ bits) 
can be achieved by storing for each vertex $v_i$ in the walk the next outgoing edge
(when $d$ is not a power of $2$ we can use arithmetic coding to compress the entire sequence jointly 
with $1$ bit loss in space).  
In either case, such encoding does not enable random access: retrieving $v_i$ requires ``unfolding" the walk 
all the way to the start, taking $t\sim \Omega(n)$ time in the worst case (by preprocessing how to unfold consecutive 
$\lg n$ steps, the decoding time can be improved to $\Omega(n/\lg n)$). The other extreme is to store 
the walk explicitly (storing $V_i$ using $\lceil \lg |G| \rceil$ bits), facilitating $O(1)$ decoding time, at the price 
of highly redundant space (at least the LHS of \eqref{eq_RW_entropy}). Of course, it  is possible to combine the 
two approaches and compress each $t$-step subwalk optimally, losing at most $\sim \lg (|G|/d)$ bits per block, 
so that decoding time is $O(t/w)$ and storage is suboptimal by an additive term of $r=O_{|G|}(n/t)$ bits, where $r$ 
is the \emph{redundancy} of the data structure. This linear tradeoff implies, for example, that if 
we wish to decode each vertex in constant time, the redundancy of the data structure is $r \sim \frac{n\lg(|G|/d)}{\lg n}$ 
on a word-RAM with word-size $w = \Theta(\lg n)$. 

We show that at the price of $r=3$ extra bits of redundancy, each vertex in the walk can 
be decoded in constant time, when the underlying graph is $d$-regular for any $d$ 
(not necessarily a power of $2$): 

\begin{theorem}[Informal]\label{thm_regular_informal}
	Given a walk $(v_0,\ldots,v_n)$ in $G$, there is a succinct \emph{cell-probe} data structure using 
	$\lg |G| +n\lg d + 3$ \emph{bits} of memory, supporting retrieval of any vertex $v_i$ ($i\in [n]$)
	in constant time, assuming word-size $w\geq \Omega(\lg n)$. Moreover, the data structure can be 
	implemented on a word RAM, assuming a precomputed\footnote{This lookup table only contains 
	precomputed information about the graph $G$ and does not depend
	on the input walk $(v_0,\ldots,v_n)$.
	If we are storing multiple walks, it can be shared across instances, hence has a small 
	``amortized" cost.} $\ell$-bit look-up table, supporting vertex retrievals in $O(\frac{\lg\lg n}{\lg\lg \ell})$ time. 
\end{theorem}  

Dealing with general (directed, non-regular) graphs is fundamentally more challenging, the main reason 
being that: (a) the stationary distribution of non-regular graphs is not uniform; (b) the number of sub-walks 
between any two vertices $(u,v)$ is not fixed when the graph is 
non-regular or directed, i.e., the distribution on sub-walks is again nonuniform. (We elaborate 
why uniformity is crucial for succinct solutions and  
how we resolve this challenge in the following Section \ref{sec_tech_overview}). 
Intuitively, this means that if we want our space to (roughly) match the entropy of the walk,  
then for non-regular graphs \emph{variable-length} coding must be used, making decoding in constant-time 
very challenging (as the memory locations corresponding to substrings become unpredictable). 

To this end, it is natural to consider two space benchmarks for general graphs. The first benchmark is the 
\emph{point-wise} optimal space for storing the walk $v = (v_0,\ldots, v_n)$, i.e., its 
empirical entropy 
\begin{align} \label{eq_pointwise_benchmark}
\lg |G|+\sum_{i=0}^{n-1} \lg (\deg_{_G}(v_i)),
\end{align}
which is clearly the information-theoretic minimum space (as the expectation of this term over a 
random walk $V$ is precisely the Shannon entropy of the walk). A somewhat more modest 
\emph{worst-case} space benchmark (depending on how ``non-regular" the graph $G$ is), is 
\begin{align} \label{eq_wc_benchmark}
\lg \kappa(G,n) = \lg (\mathbf{1}^{\top} {A_G}^n \mathbf{1}) 
\end{align}
where $\kappa(G,n)$ is the number of length-$n$ walks on $G$ and $A_G$ is the adjacency matrix of $G$. 
Note that for regular graphs the two benchmarks \eqref{eq_pointwise_benchmark} and \eqref{eq_wc_benchmark} 
are the same, hence Theorem \ref{thm_regular_informal} is best possible. 

Our first result for general graphs, building on the work of \cite{Pat08}, is a data structure with 
$O(1)$ extra bits beyond the point-wise optimal space, and $O(\lg n)$ decoding 
time on a word-RAM.  

\begin{theorem}[Point-wise compression for general graphs, Informal] \label{thm_pointwise_general_informal}
	Given a length-$n$ walk $(v_0,\ldots,v_n)$ on any directed graph $G$, 
 	there is a succinct data structure that uses $\lg_2|G|+\sum_{i=0}^{n-1} \lg (\deg(v_i)) + O(1)$ bits of space, 
      and retrieves any vertex $v_i$ in $O(\lg n)$ time on a word-RAM, assuming 
      a precomputed\footnote{Once again, this lookup table is independent of the walk and only contains 
      precomputed information about $G$ for decoding, hence its ``amortized" space cost is small as 
      it is shared across instances.} lookup-table of size $\ell = \tilde{O}(n^6)$. 
\end{theorem}

Nevertheless, we show that if one is willing to match the \emph{worst-case} space benchmark 
\eqref{eq_wc_benchmark}, then \emph{constant} query time is possible for any graph $G$.  
This is our main result  and the technical centerpiece of the paper.  We state this result in its most general form.

\begin{theorem}[Worst-case compression for general graphs, Informal] \label{thm_wc_general_informal}
	Given a length-$n$ walk 
	on any strongly connected directed graph $G$, 
      there is a data structure that uses $\lg (\mathbf{1}^{\top} A^n \mathbf{1})+O(1)$ bits of space,    
      and retrieves any vertex $v_i$ in $O(\frac{\lg\lg n}{\lg\lg \ell})$ time on a word-RAM, assuming 
      a precomputed lookup table of size $\ell$. 
      For general directed 
      graphs, the same holds albeit with $O(\lg n)$ bits of redundancy.
\end{theorem}



It is natural to ask whether constant-time decoding is possible even with respect to the point-wise space benchmark.
Our final result shows that any improvement on Theorem \ref{thm_pointwise_general_informal} 
would lead to an improvement on the long-standing \emph{succinct Dictionary} problem 
\cite{Pagh02, Pat08}: We present a succinct reduction from Dictionary 
to storing a random walk on some non-regular (directed) graph, under the technical restriction 
that all marginal symbol frequencies in the input dictionary string 
are powers of two  
\footnote{This seems a minor condition when dealing with \emph{arbitrary} (nonuniform) priors $\mu$, 
as the ``hard" part is variable-length coding.}. 

\begin{theorem}\label{thm_LB_informal}
Let  $D$ be a succinct cell-probe data structure for storing a walk $(v_1,\ldots,v_n)$ over general (directed) 
graphs $G$, using $\sum_i \lg(\deg(v_i)) + r$ bits of space, and query time $t$ for each $v_i$. 
Then for any constant-size alphabet $\Sigma$, there is a succinct dictionary storing $x\in \Sigma^n$, 
with space $H_0(x) + r$ bits and query time query time $t$,  where $H_0(x)$ is the \emph{zeroth-order} 
empirical entropy of $x$. This reduction holds for any input string  $x$ with 
empirical frequencies which are inverses of powers of $2$ lower bounded by a constant.  
\end{theorem} 

This reduction formalizes the intuition that the 
bottleneck in both problems is \emph{variable-length} coding (unlike the case of regular graphs), 
and provides evidence that Theorem \ref{thm_pointwise_general_informal} might be optimal
(Indeed, in the bit-probe model ($w=1$),  it is known that any succinct Dictionary with constant 
or even $r=O(\lg n)$ redundancy, must have $\Omega(\lg n)$ decoding time \cite{Vio09}. For 
the related \emph{succinct Partial Sums}  problem \cite{PV09}, an $\Omega(\lg n)$ lower bound 
holds even in the cell-probe model $(w=\Theta(\lg n))$).

\begin{remark}
It is noteworthy that the assumption throughout the paper that the underlying graph $G$ is of finite size 
is necessary: A length-$n$ random walk on the undirected $n$-\emph{cycle} is equivalent to the 
succinct \emph{Partial Sums} problem, for which there is a cell-probe lower bound of $t\geq \Omega(\lg n/\lg\lg n)$ 
for constant redundancy \cite{PV09}. A key feature that circumvents this lower bound 
in our proofs 
is the fact that the walk \emph{mixes} fast (which doesn't happen on the $n$-cycle). 
This justifies the restriction to fixed  sized graphs. 
\end{remark}


\subsection{Technical Overview} \label{sec_tech_overview}
At a high level, our data structures use a \emph{two-level} dictionary scheme to encode the random walk.
Let us first focus on $d$-regular graphs $G$.
To store a walk on $G$, we begin by storing the set of vertices $V_i$ that are $\Theta(\lg n)$-far apart in the walk (called the \emph{milestones}) using the succinct dictionary of \cite{DPT10} (which we will refer to as 
\emph{the \DPT{} dictionary} in the following).
The \DPT{} dictionary is able to store any string $z\in \Sigma^k$ with only constant redundancy ($k\lg|\Sigma|+O(1)$ bits in total) and allow one to retrieve each $z_i$ in constant time (Theorem \ref{thm_DPT}).
In particular, when the string $z$ is drawn from the \emph{uniform} distribution, its space usage matches the input entropy.
When $G$ is regular (and hence has a uniform stationary distribution), this is indeed the case, as a standard mixing-time argument implies that the milestones are very close to being \emph{uniform and independent}. 

The second important feature of milestones is that they break the dependence across the remaining vertices in the walk.
That is, the milestones partitions the walk into ``blocks'' $B_j$ of length $\Theta(\lg n)$.
The Markov property of a random walk implies that these blocks of vertices are independent conditioned on the milestones.
Thus, the next idea is to use another (separate) dictionary to encode the vector of intermediate blocks 
$(B_1,\ldots, B_{O(n/\lg n)}) \in w^{O(n/\lg n)}$, \emph{conditioned} on the milestones.
Again because of the mixing-time argument, for each block $B_j$, the number of possible subwalks in the block \emph{given} the two milestones $V_i$ and $V_{i+O(\lg n)}$ is always approximately $d^{|B_i|}/|G|$, \emph{regardless of} the actual values that $V_i$ and $V_{i+O(\lg n)}$ take.
Hence, one can encode each subwalk using an integer no larger than $|\Sigma|=(1+o(1))d^{|B_i|}/|G|$, and the second dictionary is used to succinctly store these integer encodings.
When the input is uniformly random, these integer encodings are uniform and independent, hence \DPT{} matches the entropy.
Also, note that since each block $B_j$ is of length $O(\lg n)$, it fits in a constant number of words (as $|G|=O(1)$).
To see why a vertex between two milestones $V_i$ and $V_{i+\Theta(\lg n)}$ can be retrieved efficiently, note that the \DPT{} dictionary allows us to retrieve each symbol in the vector of milestones in constant time.
Therefore, it suffices to retrieve $V_i$, $V_{i+\Theta(\lg n)}$ as well as the block $B_j$ between them using a 
constant number of probes to both dictionaries.
This gives us enough information to recover the entire subwalk from $V_i$ to $V_{i+\Theta(\lg n)}$, and in particular, to recover the queried vertex.
Although the above argument assumes a uniformly random input, we emphasize that our data structure works for \emph{worst-case} input and queries.

Dealing with general graphs is a different ballgame, and the above construction does not obtain the 
claimed worst-case space bound.
The main reason is that for a uniformly sampled random length-$n$ walk, the marginal distribution of each (milestone) 
$V_i$ may be arbitrary and non-uniform.
Hence, we cannot apply the \DPT{} dictionary directly on the milestone-vector, as for non-uniform vectors its space 
usage would be much higher than the input entropy.
To overcome this issue, we use an idea inspired by rejection-sampling: 
We consider not only each milestone, but also the subwalks near it (with close indices).
We partition the set of subwalks into \emph{bundles}, such that for a uniformly random input, the bundle that 
contains the given subwalk is uniform.
More precisely, for each milestone $V_i$, 
\begin{enumerate}
\item a bundle is a subset of length-$2l$ subwalks that the input $(V_{i-l},V_{i-l+1},\ldots,V_{i+l})$ can take the value, for $l=\Theta(\lg n)$;
\item all subwalks in the same bundle have an \emph{identical} middle vertex $V_i = v_i$ for some $v_i \in [|G|]$;
\item each bundle consists of approximately the same number of subwalks, i.e., for a uniformly random input, the bundle that contains it is roughly uniform;
\item for different $V_i$, the bundles are almost ``independent.''
\end{enumerate}
We prove the existence of such good partition to subwalks using a spectral argument which helps control the number of 
subwalks  delimited by any fixed pair of vertices $(u,v)$. 
Given such bundling, instead of using \DPT{} to store the milestones themselves, 
we use it to store the \emph{name of the bundle} near each milestone.
By Item 3 and 4 above, when the walk is uniformly random, the bundles to be stored are \emph{uniform and independent}.
Hence, the size of \DPT{} dictionary matches the input entropy, as desired.
The second part of the data structure is similar to the regular graph case at the high level.
We store the blocks between the consecutive milestones, but now conditioned on \emph{the bundles} (not the milestones).
The query algorithm is also similar: to retrieve a vertex between two consecutive milestones, we first retrieve the bundles that contain the subwalks near them, and then retrieve the block, which is encoded conditioned on the bundles.
The actual construction for the non-regular case is much more technical than regular graphs.
See Section~\ref{sec_upper} for details.


The final challenge in our scheme  
is implementing the decoding process on a RAM -- The above scheme only provides a \emph{cell-probe} data structure 
(assuming arbitrary operations on $O(\lg n)$-bit words), since the aforementioned encoding of blocks $+$ bundles is extremely 
implicit. Therefore, decoding with standard RAM operations appears to require storing giant lookup tables to allow efficient retrieval. 
Circumventing this intuition is one of our main contributions.  
It turns out that the basic problem that needs to be solved is the following: Given a walk $(v_0,\ldots,v_l)$ of length $l$, for $l= O(\lg n)$, from $x$ to $y$ ($v_0=x$ and $v_l=y$), encode this walk using an integer between 1 and the number of such walks $({e_x}^{\top} (A_G)^l e_y)$, such that given an index $i\in[0, l]$, one can decode $v_i$ efficiently.
Note that both two endpoints $x$ and $y$ are given, and do not need to be encoded.
They would ultimately correspond to the milestones, which have already been stored elsewhere.
Our encoding procedure is based on a $B$-way divide-and-conquer.
As an example, when $B=2$, we first recursively encode $(v_0,\ldots,v_{l/2})$ and $(v_{l/2},\ldots,v_l)$, and then ``merge'' the halves.
Given the encoding of the two halves, the final encoding of the entire walk is the \emph{index in the lexicographic order} of the triple: I) $v_{l/2}$, II) encoding of $(v_0,\ldots,v_{l/2})$, and III) encoding of $(v_{l/2},\ldots,v_l)$.
To decode $v_i$, we first decode the entire triple by enumerate all possible values for $v_{l/2}$, and count how many length-$l$ walks have this particular value for $v_{l/2}$ (these counts can be stored in the lookup table).
Then we recurse into one of the two halves based on the value of $i$.
This gives us an $O(\lg l)=O(\lg\lg n)$ time decoding algorithm.
We can generalize this idea to larger $B$, by recursing on each of the $1/B$-fraction of the input.
However, the decoding algorithm becomes more complicated, as we cannot afford to recover the entire $O(B)$-tuple.
It turns out that to efficiently decode, one will need to do a \emph{predecessor search} on a set of $\exp(B)$ integers with \emph{increasing gaps}.
This set depends only on the underlying graph $G$, therefore, we can store this predecessor search data structure in the lookup table, taking $\exp(B)$ space.
Now, the recursion only has $\lg_B l=\frac{\lg\lg n}{\lg B}$ levels, each level still takes constant time, obtaining the claimed tradeoff.

It is worth noting that the \DPT{} dictionary supports \emph{appending} extra symbols in constant amortized time, 
hence it supports the (append-only) \emph{online} version of the problem. 
Since our data structure consists of two instances of \DPT{}, when the vertices in the random walk is given one at a time, 
our random-walk data structure can be build \emph{online} as well, with amortized constant update time.
The sizes of the two instances of \DPT{} may increase over time, but their \emph{ratio} remains fixed.
By storing memory words of the two instances in \emph{interleaving} memory locations   
with respect to the right ratio, we will not need to relocate the memory words when new vertices arrive.

For the harder task for matching the information-theoretic minimum space (i.e., the point-wise empirical entropy of the walk 
$\sum_i \lg(\deg(v_i))$), we show how to efficiently 
encode the random-walk problem (on arbitrary constant-size graphs) using P\v{a}tra\c{s}cu's \emph{aB-trees}~\cite{Pat08}. 
This provides a constant-redundancy scheme but increases decoding time to $t=O(\lg n)$.  
We provide evidence that this blowup might indeed be 
necessary: We design a \emph{succinct reduction} showing how to ``embed" the classic \emph{Dictionary} problem 
on inputs $x\in \mu^n$, as a random-walk on the \emph{Huffman Tree} of $\mu$, augmented with certain 
directed paths to enable fast decoding (this works for any distribution $\mu$ where each $\mu(i)$ is the inverse of a 
power of 2). This reduction shows that any asymptotic improvement on the aforementioned random walk data 
structure would lead to a improvement on the long-standing  dictionary problem 
(see  Theorem \ref{thm_Dic_to_nonreg_RW}).


\subsection{Related work} \label{sec_related_work}
The ``prior-free" version of the (generic) LDSC problem has been studied in the  
pattern-matching and information-theory communities, where in this setting the compression 
benchmark is typically some (high-order) empirical entropy $H_k(x)$ capturing 
``bounded-correlations" in the text (see e.g. \cite{Manzini99} for a formal definition). 
It has been shown that classical  \emph{universal compression}  schemes (such as Lempel-Ziv \cite{LZ78} 
compression as well as the Burrows-Wheeler transform \cite{BW}) can be made to have random-access to 
individual symbols at the price of small (but not \emph{succinct}) loss in compression 
(e.g., \cite{FM, DLRR13, TBW18, SW18}).  
In general, these results are incomparable to our distributional setting of the LDSC problem, 
as the space analysis is asymptotic in nature and based on the (ergodic) assumption that the input 
$(X_1,\ldots, X_n)$ has finite-range correlations or restricted structure (Indeed, this is 
confirmed by our impossibility result for general LDSC in Corollary \ref{cor_LDSC_LBs}). 
In the case of $k$-order Markov chains, where the $k$'th empirical and Shannon entropies actually match  
($H_k(X) = H(X)$), all previously known LDSC schemes generally have redundancy at least 
$r \geq \Omega(n/\lg n)$ regardless of the entropy of the source and decoding time 
(see e.g., \cite{DLRR13, TBW18} and references therein).  
In contrast, our schemes for achieve $r=O(1)$ redundancy and constant (or at most logarithmic) 
query time on a RAM.   

Technically speaking, the most  relevant literature to our paper is the work on 
\emph{succinct Dictionary} problem \cite{Buhrman02, Pagh02, Raman07_succinct, GolynskiRR08, 
Pat08, DPT10} (see Section \ref{subsec_dictionary_vs_RW} for the formal problem definition). 
The state-of-art, Due to \Pat \cite{Pat08}, is a 
succinct data structure (LDSC) for \emph{product distributions} 
on $n$-letter strings $\overline{X}\sim \mu^n$, with an \emph{exponential} tradeoff between 
query time and redundancy $r = O(n/(\lg n/t)^t)$ over the expected (i.e., zeroth-order) entropy 
of $\overline{X}$.  Whether this tradeoff is optimal is a long-standing open problem in succinct data structures 
(more on this in Section \ref{subsec_dictionary_vs_RW}). 
Interestingly, if $\mu$ is the \emph{uniform} distribution, a followup work of Dodis et. al 
\cite{DPT10} showed that constant redundancy and decoding time is possible on a word-RAM.  
(see Theorem \ref{thm_DPT}, which will also play a key role in our data structures). 
In some sense, our results show that such optimal tradeoff  
carries over to strings with finite-length correlations. 


\section{Preliminaries}

\subsection{Graphs and Random walks.} 
Let $G$ be an unweighed graph, $A$ be its adjacency matrix, and $P$ be the transition matrix of the random walk on $G$.
$P$ is $A$ with every row normalized to sum equal to 1.
\paragraph{Undirected regular graphs.}
When $G$ is undirected and $d$-regular, $A$ is real symmetric, and $P=\frac{1}{d}A$. 
All of the eigenvalues are real.
In particular, the largest eigenvalue of $A$ is equal to $d$, with eigenvector $\bOne$, the all-one vector.
Moreover, if $G$ is connected and non-bipartite, then all other eigenvalues have absolute values strictly less than $d$.
Suppose all other eigenvalues are all at most $(1-\epsilon)d$, the following lemma is known.
\begin{lemma}\label{lem_undirect_mix}
	Let $X$ be a vector of dimension $|G|$ corresponding to some distribution over the vertices.
	Let $U$ be the vector corresponding to the uniform distribution.
	We have
	\[
		\|\frac{1}{d}\cdot A(X-U)\|_2\leq (1-\epsilon)\|X-U\|_2.
	\]
\end{lemma}
That is, every step of the random walk on $G$ makes the distribution $X$ close to the uniform by a constant factor.

\paragraph{Strongly connected aperiodic graphs.} 
For directed graph $G$, it is strongly connected if every node can be reached from every other node via directed paths.
For strongly connected $G$, it is aperiodic if the greatest common divisor of the lengths of all cycles in $G$ is equal to 1. 
One may view strongly connectivity and aperiodicity as generalization of connectivity and non-bipartiteness for undirected graphs from above.

For strongly connected aperiodic $G$, the Perron--Frobenius theorem asserts that
\begin{itemize}
	\item let $\lambda$ be the spectral radius of $A$, then $\lambda$ is an eigenvalue with multiplicity 1;
	\item let $\pi^{\top}$ and $\sigma$ be its left and right eigenvectors with eigenvalue $\lambda$ respectively, all coordinates of $\pi$ and $\sigma$ are positive.
\end{itemize}
Similarly, for the transition matrix $P$, 
\begin{itemize}
	\item it has spectral radius $1$, and $1$ is an eigenvalue with multiplicity 1;
	\item let $\nu^{\top}$ be its left eigenvector, then $\nu^{\top}$ is the unique stationary distribution.
\end{itemize}
Moreover, let $X_1$ and $X_2$ be two vectors, we have the following approximation on $X_1^{\top}A^lX_2$ (e.g., see~\cite{fill1991}).
\begin{lemma}\label{lem_direct_mix}
	We have
	\[
		X_1^{\top}A^lX_2=\lambda^l\cdot \left(\frac{\left<\sigma,X_1\right>\left<\pi,X_2\right>}{\left<\sigma,\pi\right>}\pm O(\|X_1\|_2\|X_2\|_2)\cdot \exp(-\Omega(l))\right).
	\]
\end{lemma}
\begin{proof}[Proof (sketch)]
	Let us first decompose $X_2$ into two vectors $\alpha\cdot\sigma$ and $X'$ such that $\left<\pi,X'\right>=0$.
	Thus, we have
	\[
		\alpha=\frac{\left<\pi,X_2\right>}{\left<\sigma,\pi\right>},
	\]
	and
	\[
		X'=X_2-\alpha\cdot\sigma.
	\]

	For the first term, we have $A^l(\alpha\cdot\sigma)=\lambda^l\alpha\cdot \sigma$.
	To estimate the second term, let $D=\mathrm{diag}(\sigma/\pi)$, i.e., $D$ is the diagonal matrix with the $i$-th diagonal entry equal to $\sigma_i/\pi_i$.
	Since $A$ is strongly connected and aperiodic, there exists a constant $c$ such that all entries in $A^c$ are positive.
	We will show that $\lambda^{-2l}\cdot X'^{\top}(A^{\top})^l D^{-1}A^lX'$ decreases exponentially.
	To this end, consider the matrix
	\[
		D^{1/2}(A^{\top})^cD^{-1}A^c D^{1/2}.
	\]
	This matrix is positive real symmetric, and $\pi^{\top}D^{1/2}$ is its left eigenvector with eigenvalue $\lambda^{2c}$.
	Since $\pi^{\top}D^{1/2}$ is positive, by the Perron--Frobenius theorem, all other eigenvalues have magnitude strictly less than $\lambda^{2c}$, and assume they are all at most $((1-\epsilon)\lambda)^{2c}$ for some $\epsilon>0$.

	For every $i$, since $\left<\pi^{\top}D^{1/2},D^{-1/2}A^iX'\right>=0$, we have
	\[
		X'^{\top}(A^{\top})^{i+c}D^{-1}A^{i+c} X'\leq ((1-\epsilon)\lambda)^{2c}\cdot X'^{\top}(A^{\top})^{i}D^{-1}A^{i} X'.
	\]
	Therefore, we have
	\begin{align*}
		X'^{\top}(A^{\top})^l D^{-1}A^lX'&\leq \exp(-\Omega(l))\cdot \lambda^{2l}X'^{\top}D^{-1}X'\\
		&\leq \exp(-\Omega(l))\cdot \lambda^{2l}\cdot O(\|X'\|_2^2) \\
		&\leq \exp(-\Omega(l))\cdot \lambda^{2l}\cdot O(\|X_2\|_2^2).
	\end{align*}
	Thus, $\|A^lX'\|_2\leq \lambda^l\cdot O(\|X_2\|_2)\cdot \exp(-\Omega(l))$.

	Combining the two parts, we have
	\[
		X_1^{\top}A^lX_2=\frac{\left<\sigma,X_1\right>\left<\pi,X_2\right>}{\left<\sigma,\pi\right>}\cdot\lambda^l\pm \lambda^l\cdot O(\|X_1\|_2\|X_2\|_2)\cdot \exp(-\Omega(l)).
	\]
	This proves the lemma.
\end{proof}

One can also prove a similar statement about $P$.
\begin{lemma}\label{lem_mix_P}
	$X_1^{\top}P^lX_2=\left<\bOne,X_1\right>\left<\nu,X_2\right>\pm O(\|X_1\|_2\|X_2\|_2)\cdot \exp(-\Omega(l)).$
\end{lemma}

\subsection{Space Benchmarks}
There are $\bOne^{\top}A^n\bOne$ different walks $(v_0,\ldots,v_n)$ on $G$ of length $n$.
Thus, $\lg \bOne^{\top}A^n\bOne$ bits is the optimal \emph{worst-case} space for storing a $n$-step random walk.

However, for non-regular graphs, some walk may appear with a higher probability than the others.
By Lemma~\ref{lem_mix_P}, the marginal of $V_i$ quickly converges to $\nu$.
Therefore, we have
\[
	H(V_{i+1}\mid V_i)=\sum_x \nu_x\cdot \lg\deg(x)+\exp(-\Omega(i)).
\]
By chain-rule and the Markov property of a random walk, we have
\[
	H(V_0,\ldots,V_n)=n\cdot \left(\sum_x \nu_x\cdot \lg\deg(x)\right)+O(1).
\]
This is the optimal \emph{expected} space any data structure can achieve.

Finally, note that the \emph{point-wise} space benchmark defined in the introduction
\[
	\lg |G|+\sum_{i=0}^{n-1}\lg\deg(v_i)
\]
implies almost optimal expected, since it assigns $\lceil\log 1/p\rceil$ bits to a walk that appears with probability $p$.

\subsection{Word-RAM model and Succinct Data Structures}

The following surprising (yet limited) result of \cite{DPT10} will be used as a key subroutine in 
our data structures. 

\begin{theorem}[Succinct dictionary for \emph{uniform} strings, \cite{DPT10}] \label{thm_DPT}
For any $|\Sigma| \leq w$ (not necessarily a power of 2), there is a succinct Dictionary 
storing any string $x\in \Sigma^n$ using $\lceil n\lg |\Sigma| \rceil$ \emph{bits} of space, 
supporting constant-time retrieval of any $x_i$, on a word-RAM with word size $w=\Theta(\lg n)$, 
assuming pre-computed lookup tables of $O(\lg n)$ words. 
Moreover, the dictionary supports the online (``append-only") version with constant update and query times. 
\end{theorem}

We remark that the cost of the pre-computed lookup tables is negligible and they are shared across 
instances (they are essentially small code-books for decoding and hence do not depend on the input 
itself).  This is a standard requirement for variable-length coding schemes.


\section{Succinct Data Structures for Storing Random Walks}\label{sec_upper}
\subsection{Warmup: $d$-Regular Graphs}\label{sec_upper_reg}

Fix a $d$-regular graph $G$ with $k$ vertices, we assume that $G$ is connected and non-bipartite.
Denote its adjacency matrix by $A$.
In the following, we show that a walk on $G$ can be stored succinctly, and allow efficient decoding of each vertex.

\begin{theorem}
	Given a walk $(v_0,\ldots,v_n)$ in $G$, there is a succinct \emph{cell-probe} data structure that uses at most $\lg_2 |G|+n\lg_2 d+3$ bits of memory, and supports retrieving each $v_q$ in constant time for $q=0,\ldots,n$, assuming the word-size $w\geq \Omega(\lg n)$.
	Moreover, the data structure can be implemented on a word RAM using an extra $r$-bit lookup table, which depends only on $G$, supporting vertex retrievals in $O(\frac{\lg\lg n}{\lg\lg r})$ time, for any $r\geq \Omega(\lg^2 n)$.
\end{theorem}
\begin{proof}
	The idea is to divide the walk into blocks of length $l$ for $l=\Theta(\lg n)$, so that if we take every $l$-th node in a uniformly random walk (which we call the milestones), they look almost independent.
	We first store all these milestones using the \DPT{} dictionary, which introduces no more than one bit of redundancy.
	Then note that conditioned on the milestones, the subwalk between any two adjacent milestones are also independent and uniform over some set, then we store the subwalks \emph{conditioned on the milestones} using \DPT{} again.

	More formally, the top eigenvalue of $A$ is equal to $d$.
	Suppose all other eigenvalues are at most $(1-\epsilon)d$, we set $m=\lfloor\frac{\epsilon}{2}\cdot n/\ln n\rfloor$, and $l=n/m\geq \frac{2}{\epsilon}\ln n$.
	We divide the walk into $m$ blocks of length approximately $l$ each.
	Let $a_i=\lfloor (i-1)\cdot l\rfloor$ for $1\leq i\leq m+1$, be the $m+1$ milestones.
	Note that $a_1=0$ and $a_{m+1}=n$.
	The $i$-th block is from the $i$-th milestone $v_{a_i}$ to the $(i+1)$-th $v_{a_{i+1}}$.
	Hence, the length of the block is in $(l-1, l+1)$.

	The first part of the data structure stores all milestones, i.e., $v_{a_i}$ for all $1\leq i\leq m+1$ using \DPT.
	This part uses $\lceil (m+1)\lg_2 |G|\rceil$ bits of memory.

	Next, we store the subwalks between the adjacent milestones.
	By Lemma~\ref{lem_undirect_mix}, after $l$ steps of random walk from any vertex $v_{a_i}$, the $\ell_2$ distance to the uniform distribution $U$ is at most $(1-\epsilon_k)^l\leq n^{-2}$, hence the $\ell_{\infty}$ distance to $U$ is also at most $n^{-2}$.
	In particular, the probability that the random walk ends at each vertex is upper bounded by $1/|G|+1/n^2$.
	That is, given the milestones $v_{a_i}$ and $v_{a_{i+1}}$, the number of possible subwalks in the $i$-th block (from $v_{a_i}$ to $v_{a_{i+1}}$) is always at most
	\[
		\left(\frac{1}{|G|}+\frac{1}{n^2}\right)\cdot d^{a_{i+1}-a_i},
	\]
	\emph{regardless of} the values of $v_{a_i}$ and $v_{a_{i+1}}$.

	Hence, the subwalk can be encoded using a positive integer between $1$ and $\lfloor(1/|G|+1/n^2)\cdot d^{a_{i+1}-a_i}\rfloor$.
	Note that since $l=O(\lg n)$, this integer has $O(w)$ bits.
	For cell-probe data structures, one can hardwire an arbitrary encoding for every possible pair of milestones and every possible length of the block. 
	We will defer the implementation on RAM to the end of this subsection (Lemma~\ref{lem_ram_decode}), and let us focus on the main construction for now.

	We obtain an integer for each subwalk between adjacent milestones.
	We again use \DPT{} to store these integers.
	This part uses at most 
	\[
		\left\lceil\sum_{i=1}^{m} \lg_2 \left(\left(1/|G|+1/n^2\right)\cdot d^{a_{i+1}-a_i}\right)\right\rceil
	\]
	bits of memory.

	Therefore, the number of bits we use in total is at most
	\begin{align*}
		&\ ((m+1)\lg_2 |G|+1)+(\sum_{i=1}^{m} \lg_2 \left(\left(|G|^{-1}+n^{-2}\right)\cdot d^{a_{i+1}-a_i}\right)+1)\\
		=&\ (m+1)\lg_2 |G|+\sum_{i=1}^{m}(a_{i+1}-a_i)\lg_2 d+m \lg_2(|G|^{-1}+n^{-2})+2 \\
		=&\ \lg_2 |G|+(a_{m+1}-a_1)\lg_2 d+m\lg_2 (1+|G|\cdot n^{-2})+2 \\
		\leq&\ \lg_2 |G|+n\lg_2 d+m|G|\cdot n^{-2}+2 \\
		\leq&\ \lg_2 |G|+n\lg_2 d+3.
	\end{align*}

	The query algorithm in the cell-probe model is straightforward.
	To retrieve $v_q$, we first compute the block it belongs to: $i=\lceil q/l\rceil$.	
	Then we query the first part of the data structure to retrieve both $v_{a_i}$ and $v_{a_{i+1}}$, and query the second part to retrieve the integer that encodes the subwalk between them conditioned on $v_{a_i}$ and $v_{a_{i+1}}$.
	They together recover the whole subwalk, and in particular, $v_q$.

\end{proof}


\paragraph{Decoding on the word-RAM.}

Denote by $\cN_l(x, y)$, the number of different walks from $x$ to $y$ of length $l$, i.e., $\cN_l(x, y)=\bE_x^\top A^l \bE_y$.
In the following, we show that for every $x,y$ and $l=O(\lg n)$, there is a way to encode such a walk using an integer in $[\cN_l(x, y)]$, which allows fast decoding.
\begin{lemma}\label{lem_ram_decode}
	For $l=O(\lg n)$, given a length-$l$ walk from $x$ to $y$, one can encode it using an integer $K\in[\cN_l(x, y)]$ such that with an extra lookup table of size $r$, 
	depending only on the graph $G$, one can retrieve the $q$-th vertex in the walk in $O(\frac{\lg\lg n}{\lg\lg r})$ time, for any $r\geq \Omega(\lg^2 n)$.
\end{lemma}

\begin{proof}
	To better demonstrate how we implement this subroutine, we will first present a solution with a slower decoding time of $O(\lg\frac{\lg n}{\lg r})=O(\lg\lg n-\lg\lg r)$, and for simplicity, we assume $l$ is a power of two for now.

	Given a length-$l$ walk from $x$ to $y$, the encoding procedure is based on divide-and-conquer.
	We first find the vertex $z$ in the middle of the walk, and recursively encode the two length-$l/2$ walks from $x$ to $z$ and from $z$ to $y$ (given the endpoints).
	Suppose from the recursion, we obtained two integers $K_1\in[\cN_{l/2}(x,z)]$ and $K_2\in[\cN_{l/2}(z,y)]$, then the final encoding will be the index in the lexicographic order of the triple $(z, K_1, K_2)$.
	That is, we encode the walk by the integer
	\[
		K=\sum_{z'<z} \cN_{l/2}(x,z')\cN_{l/2}(z',y)+(K_1-1)\cN_{l/2}(z, y)+K_2.
	\]
	Hence, $K$ is an integer between $1$ and $\sum_{z}\cN_{l/2}(x,z)\cN_{l/2}(z,y)=\cN_l(x,y)$.

	We store the constants $\cN_{l'}(x,y)$ for all $l'\in[1,l]$ and vertices $x, y$ in the lookup table.
	The decoding procedure is straightforward.
	Given $x, y$ and $K$, to decode the $q$-th vertex in the walk, we first recover the triple $(z, K_1, K_2)$.
	To this end, we cycle through all $z$ in the alphabetic order: If $K>\cN_{l/2}(x,z)\cN_{l/2}(z,y)$, subtract $K$ by $\cN_{l/2}(x,z)\cN_{l/2}(z,y)$, and increment $z$; Otherwise, the current $z$ is the correct middle vertex.
	Then $K_1$ and $K_2$ can be computed by $K_1=\lfloor (K-1)/\cN_{l/2}(z,y)\rfloor+1$ and $K_2=(K-1)\!\!\mod \cN_{l/2}(z,y)+1$.
	Next,
	\begin{itemize}
		\item if $q=l/2$, $z$ is the queried vertex, and return $z$;
		\item if $q<l/2$, recursively query the $q$-th vertex in the length-$l/2$ walk from $x$ to $z$;
		\item if $q>l/2$, recursively query the $(q-l/2)$-th vertex in the walk from $z$ to $y$.
	\end{itemize}
	This gives us a decoding algorithm running in $O(\lg\lg n)$ time.
	To obtain a faster decoding time of $t<\lg\lg n$, we could run the above recursion for $t$ levels, and arrive at a subproblem asking to decode a length-$(l/2^t)$ walk from an integer no larger than $2^{O(l/2^t)}$.
	Instead of continuing the recursion, we simply store answers to all possible such decoding subproblems in a lookup table of size $r=k^2\cdot (l/2^t)\cdot 2^{O(l/2^t)}=2^{O(2^{-t}\lg n)}$.
	Hence, the decoding time is $t=O(\lg\frac{\lg n}{\lg r})$.

	\bigskip

	To obtain the claimed decoding time of $O(\frac{\lg\lg n}{\lg\lg r})$, the encoding algorithm will be based on a $B$-way divide-and-conquer for some parameter $B>2$.
	Given a length-$l$ walk from $x$ to $y$, we first consider $B+1$ vertices $(x=)z_1,z_2,\ldots,z_B,z_{B+1}(=y)$ such that $z_i$ is the $\lfloor (i-1)l/B\rfloor$-th vertex in the walk, and recursively encode the length-$l/B$ walks from $z_i$ to $z_{i+1}$ for $i=1,\ldots,B$, obtaining integers $K_1,\ldots,K_B$.
	Next, we encode the vertices $z_2,\ldots,z_B$ by an integer $Z\in[|G|^{B-1}]$.
	The final encoding is according to the lexicographic order of the tuple $(Z,K_1,\ldots,K_B)$.
	To decode a vertex between $z_i$ and $z_{i+1}$, the algorithm will need to recover $z_i,z_{i+1}$ and $K_i$ in order to recurse.
	However, when $B$ is large, we cannot afford to enumerate $Z$ as before.
	Instead, we will use a trick from~\cite{Pat08}, which allows us to recover each $z_i$ in constant time, as well as $K_i$.

	More specifically, for each tuple $(z_2,\ldots,z_B)$, consider the number of walks that go through them at the corresponding vertices
	\[
		\prod_{i=1}^B \cN_{\lfloor i l/B\rfloor-\lfloor (i-1)l/B\rfloor}(z_i,z_{i+1}).
	\]
	In the lookup table, we store all these tuples in the \emph{sorted order by this number} (break ties arbitrarily), using $|G|^{B-1}\cdot (B-1)\lceil\lg |G|\rceil$ bits.
	The encoding $Z$ of the tuple will be the index in this order.
	Therefore, to decode $Z$ from $x,y$ and $K$, it suffices to find the largest $Z$ such that the total number of walks corresponding to a tuple $(z_2,\ldots,z_B)$ ranked prior to $Z$, is smaller than $K$.
	This is precisely a \emph{predecessor search} data structure problem over a set of size $|G|^{B-1}$ with \emph{monotone gaps}, where the set consists of for all $Z$, the total number of walks corresponding to a tuple ranked prior to $Z$, and the query is $K$.
	Because we sort the tuples by the number of corresponding walks, the numbers in the set have monotone gaps.
	It was observed in~\cite{Pat08} that the monotone gaps allow us to solve predecessor search with linear space and constant query time.
	Hence, we further store in the lookup table, this linear space predecessor search data structure, using $O(|G|^{B-1})$ words (the standard predecessor search problem only returns the number, however, it is not hard to have the data structure also return the index $Z$.)
	From the predecessor search data structure, we obtain $Z$, hence $z_i$ and $z_{i+1}$ from the first lookup table above, as well as $K'$, the rank of the tuple $(K_1,\ldots,K_B)$ given $Z$.
	Thus, $K_i$ can be computed by 
	\[
		K_i=\lfloor (K'-1)/\prod_{i'={i+1}}^B \cN_{\lfloor i' l/B\rfloor-\lfloor (i'-1)l/B\rfloor}(z_{i'},z_{i'+1})\rfloor \!\!\mod \cN_{\lfloor i l/B\rfloor-\lfloor (i-1)l/B\rfloor}(z_i,z_{i+1})+1.
	\]
	It can be computed in constant time, once we also have stored in the lookup table, for all $(z_2,\ldots,z_B)$ and all $i$, the number 
	\[
		\prod_{i'={i+1}}^B \cN_{\lfloor i' l/B\rfloor-\lfloor (i'-1)l/B\rfloor}(z_{i'},z_{i'+1}).
	\]

	Finally, the recursion has $O(\lg_B l)=O(\frac{\lg\lg n}{\lg B})$ levels, and each level takes constant time.
	We store in the lookup table for all $x,y$ and $l'\leq l$, the above tables and the predecessor search data structure using in total $r=O(|G|^{B-1}\cdot l^2)=2^{O(B)}\cdot \lg^2 n$ bits.
	Thus, for $r\geq \lg^2 n$, the decoding time is $O(\frac{\lg\lg n}{\lg\lg r})$.
	This proves the lemma.
\end{proof}




\subsection{General Graphs}\label{sec_upper_general}

In this subsection, we prove our theorem for storing a walk in a general directed graph with 
respect to worst-case space bound.
Fix a strongly connected aperiodic directed graph $G$.
Let $A$ be its adjacency matrix.
The total number of length-$n$ walks in $G$ is equal to $\mathbf{1}^{\top} A^n \mathbf{1}$.

\begin{theorem}\label{thm_general}
	Let $G$ be a strongly connected aperiodic directed graph.
	Given a length-$n$ walk $(v_0,\ldots,v_n)$ on $G$, one can construct a data structure that uses $\lg (\mathbf{1}^{\top} A^n \mathbf{1})+5$ bits of space.
	Then there is a query algorithm that given an index $q\in[0,n]$, retrieves $v_q$ in $O(\frac{\lg\lg n}{\lg\lg r})$ time with access to the data structure and a lookup table of $r\geq \Omega(\lg^2 n)$ bits, where the lookup table only contains precomputed information about $G$, on a word RAM of word-size $w=\Theta(\lg n)$.
\end{theorem}

The previous solution fails for general graph $G$.
The main reason is that if we take a uniformly random length-$n$ walk, the marginal distribution of each milestone is not uniform.
In other words, the entropy of each milestone is strictly less than $\lg |G|$.
If we store them as before, using $\lg |G|$ bits per milestone, this part of the data structure already introduces an unaffordable amount of redundancy (constant bits per milestone).

To circumvent this issue, we partition the set of subwalks near each milestone into subsets, which we refer to as the \emph{bundles}.
More specifically, let $l=\Theta(\lg n)$ be a parameter.
For each milestone $v_i$, we partition the set of all possible subwalks from $v_i$ to $v_{i+l}$ into groups $\{g_{x,j}\}$, such that
\begin{enumerate}
	\item all subwalks in $g_{x,j}$ have $v_i=x$;
	\item the size $|g_{x,j}|$ is roughly the same for all $x$ and $j$;
	\item if we sample a uniformly random walk from $g_{x,j}$, the marginal distribution of $v_{i+l}$ is roughly the same for all $x$ and $j$.
\end{enumerate}
The second property above implies that if the input walk is random, then the group $g_{x,j}$ that contains the subwalk is uniformly random.
The third property implies that which group $g_{x,j}$ contains the subwalk almost does not reveal anything about the walk after $v_{i+l}$, as the distribution already ``mixes'' at $v_{i+l}$.
We then partition the set of subwalks from $v_{i-l}$ to $v_i$ into groups $\{h_{x,j}\}$ with the similar properties.
Finally, a bundle $\cB_{x,j_1,j_2}$ will be the product set of $h_{x,j_1}$ and $g_{x,j_2}$, i.e., it is obtained by concatenating one subwalk from $h_{x,j_1}$ and one from $g_{x,j_2}$ (they both have $v_i=x$, and thus, can be concatenated).
The bundles constructed in this way have the following important properties:
\begin{enumerate}
	\item\label{item1} all subwalks in each bundle go through the same vertex at the milestone;
	\item\label{item2} each bundle consists of approximately the same number of subwalks;
	\item\label{item3} adjacent bundles are almost ``independent''.
\end{enumerate}

The data structure first stores for each milestone $v_i$, which bundle contains the subwalk near it, using \DPT{} (Item~\ref{item2} and \ref{item3} above guarantee that this part introduces less than one bit of redundancy).
Next, the data structure stores the exact subwalk between each pair of adjacent milestones, given bundles that contain (part of) it.
In the following, we elaborate on this idea, and present the details of the data structure construction.
To prove the theorem, we will use the following property about strongly connected aperiodic graphs, which is a corollary of Lemma~\ref{lem_direct_mix}.

\begin{lemma}\label{lem_power}
	Let $\lambda$ be the largest eigenvalue of $A$, $\pi^{\top}$ and $\sigma$ be its left and right eigenvectors, i.e., $\pi^{\top}A=\lambda\pi^{\top}$ and $A\sigma=\lambda\sigma$.
	Then $\sigma$ and $\pi$ both have positive coordinates.
	For large $l\geq \Omega(\lg n)$, we have the following approximation on $X_1^{\top}A^l X_2$ for vectors $X_1$ and $X_2$:
	\[
		X_1^{\top}A^l X_2=\lambda^l\cdot \left(\frac{\left<\sigma,X_1\right>\cdot\left<\pi,X_2\right>}{\left<\sigma,\pi\right>}\pm O(n^{-2}\cdot\|X_1\|_2\|X_2\|_2)\right).
	\]
\end{lemma}

\begin{proof}[Proof of Theorem~\ref{thm_general}]

	\bigskip
	Let $m=\Theta(n/\lg n)$ be an integer, such that $l=n/2m$ is large enough for Lemma~\ref{lem_power} and $\lambda^l\geq n^4$.
	Similar to the regular graph case, we divide the walk into $m$ blocks of roughly equal length.
	For $i=0,\ldots,m$, let $a_i=\lfloor 2il\rfloor$.
	The $i$-th block is from the $a_{i-1}$-th vertex in the walk to the $a_i$-th, and each $a_i$ is a milestone.
	Let $b_i$ be the midpoint between the milestones $a_i$ and $a_{i+1}$, i.e., $b_i=\lfloor (2i+1)l \rfloor$.

	Consider the set of all possible subwalks from $b_{i-1}$ to $b_i$, we will partition this set into \emph{bundles}, such that all subwalks in each bundle has the same $v_{a_i}$ (different bundles may \emph{not} necessarily have different vertex at $a_i$).
	To formally define the bundles, let us first consider the subwalk from $v_{a_i}$ to $v_{b_i}$, which has length $b_i-a_i\approx l$.
	In total, there are $\bOne^{\top}A^{b_i-a_i}\bOne$ subwalks of length $b_i-a_i$.
	For any vertices $x, y$, $\bE_x^{\top}A^{b_i-a_i}\bOne$ of them have $v_{a_i}=x$, and $\bE_x^{\top}A^{b_i-a_i}\bE_y$ have $v_{a_i}=x$ and $v_{b_i}=y$.
	Using the notation from Lemma~\ref{lem_ram_decode}, we have $\cN_{b_i-a_i}(x,y)=\bE_x^{\top}A^{b_i-a_i}\bE_y$.

	Now we partition the set of all possible subwalks from $v_{a_i}$ to $v_{b_i}$ such that $v_{a_i}=x$ into
	\[
		s_x:=\left\lfloor \frac{\bE_x^{\top}A^{b_i-a_i}\bOne}{\bOne^{\top}A^{b_i-a_i}\bOne}\cdot n^2 \right\rfloor 
	\]
	groups $g_{x,1},\ldots,g_{x,s_x}$.
	By Lemma~\ref{lem_power}, we have
	\begin{align}
		s_x&=\frac{\left<\sigma,\bE_x\right>\left<\pi,\bOne\right>\lambda^{b_i-a_i}}{\left<\sigma,\bOne\right>\left<\pi,\bOne\right>\lambda^{b_i-a_i}}\cdot n^2\left(1\pm O(n^{-2})\right) \nonumber \\
		&=\frac{\sigma_x}{\left<\sigma,\bOne\right>}\cdot n^2(1\pm O(n^{-2}))\label{eqn_sx}.
	\end{align}

	For each vertex $y$, we ensure that the number of subwalks with $v_{b_i}=y$ in each $g_{x,j}$ for $j\in[s_x]$ is approximately the same.
	More specifically, we use Lemma~\ref{lem_ram_decode} to encode the subwalk from $v_{a_i}$ to $v_{b_i}$ using an integer $K_2^{(i)}\in [\cN_{b_i-a_i}(v_{a_i},v_{b_i})]$.
	This subwalk belongs to $g_{x,j}$ for $x=v_{a_i}$ and 
	\begin{equation}\label{eqn_j2}
		j=\left\lfloor \frac{(K_2^{(i)}-1)\cdot s_x}{\cN_{b_i-a_i}(v_{a_i},v_{b_i})} \right\rfloor+1.
	\end{equation}
	Therefore, by Lemma~\ref{lem_power} and Equation~\eqref{eqn_sx} for every $x,y,j$, the number of subwalks from $x$ to $y$ in $g_{x,j}$ is at most
	\begin{align}
		\frac{\cN_{b_i-a_i}(x,y)}{s_x}+1&= \frac{\sigma_x\pi_y}{\left<\sigma,\pi\right>}\cdot \lambda^{b_i-a_i}\cdot \frac{\left<\sigma,\bOne\right>}{\sigma_x}\cdot n^{-2} (1\pm O(n^{-2})) \nonumber\\
		&= \frac{\left<\sigma,\bOne\right>\cdot \pi_y}{\left<\sigma,\pi\right>}\cdot \lambda^{b_i-a_i}n^{-2} (1\pm O(n^{-2})).\label{eqn_gxy}
	\end{align}
	Note that it is important that the $+1$ term is absorbed into the $O(n^{-2})$ term, since $\lambda^{b_i-a_i}=\Theta(\lambda^{l})\geq \Omega(n^4)$.

	Similarly, we also partition the set of subwalks from $v_{b_{i-1}}$ to $v_{a_i}$ into groups $h_{x, 1},\ldots,h_{x,t_x}$, for 
	\begin{equation}
		t_x:=\left\lfloor \frac{\bOne^{\top}A^{a_i-b_{i-1}}\bE_x}{\bOne^{\top}A^{a_i-b_{i-1}}\bOne}\cdot n^2 \right\rfloor =\frac{\pi_x}{\left<\pi,\bOne\right>}\cdot n^2(1\pm O(n^{-2})).\label{eqn_tx}
	\end{equation}
	Again, we use Lemma~\ref{lem_ram_decode} to encode the subwalk from $v_{b_{i-1}}$ to $v_{a_i}$ using an integer $K_{1}^{(i)}\in [\cN_{a_i-b_{i-1}}(v_{b_{i-1}},v_{a_i})]$.
	This subwalk belongs to $h_{x,j}$ for $x=v_{a_i}$ and 
	\begin{equation}\label{eqn_j1}
		j=\left\lfloor \frac{(K_{1}^{(i)}-1)\cdot t_x}{\cN_{a_i-b_{i-1}}(v_{b_{i-1}},v_{a_i})} \right\rfloor+1.
	\end{equation}
	The number of subwalks from $y$ to $x$ in $h_{x,j}$ is at most
	\begin{equation}\label{eqn_hxy}
		\frac{\cN_{a_i-b_{i-1}}(y,x)}{t_x}+1= \frac{\sigma_y\cdot \left<\pi,\bOne\right>}{\left<\sigma,\pi\right>}\cdot \lambda^{a_i-b_{i-1}}n^{-2} (1\pm O(n^{-2})).
	\end{equation}


	Now we describe how we partition the subwalks from $v_{b_{i-1}}$ to $v_{b_i}$ into bundles.
	Each bundle has the form $h_{x,j_1}\times g_{x,j_2}$ for some vertex $x$ and $j_1\in[t_x],j_2\in[s_x]$.
	That is, we define the bundle $\cB_{x,j_1,j_2}$ as
	\[
		\cB_{x,j_1,j_2}:=\{p_1\circ p_2: p_1\in h_{x,j_1},p_2\in g_{x,j_2}\},
	\]
	where $p_1\circ p_2$ concatenates the two paths.
	It is not hard to verify that $\left\{\cB_{x,j_1,j_2}\right\}$ is a partition, and the number of different bundles is
	\begin{equation}\label{eqn_num_bundle}
		\sum_x s_xt_x=\frac{\left<\sigma,\pi\right>}{\left<\sigma,\bOne\right>\left<\pi,\bOne\right>}\cdot n^4(1\pm O(n^{-2}))
	\end{equation}
	by Equation~\eqref{eqn_sx} and \eqref{eqn_tx}.
	In particular, the index of the bundle $(x, j_1, j_2)$ can be encoded using the integer
	\[
		\sum_{x'<x} s_{x'}t_{x'}+(j-1)t_x+j_2.
	\]

	As special cases, for $i=0$, the subwalk is from $v_{a_0}$ to $v_{b_0}$, and the bundles are defined to be $\cB_{x,j_2}:=g_{x,j_2}$ for vertex $x$ and $j_2\in[s_x]$.
	For $i=m$, the bundles are $\cB_{x,j_1}:=h_{x,j_1}$ for vertex $x$ and $j_1\in[t_x]$.
	In both cases, the number of bundles is at most $n^2$.

	\paragraph{Data structure construction and space analysis.} Now, we are ready to describe how to construct the data structure from the input $(v_0,\ldots,v_n)$.
	We first use Lemma~\ref{lem_ram_decode} to encode the subwalks from $v_{a_i}$ to $v_{b_i}$ obtaining $K_2^{(i)}$, and subwalk from $v_{a_{i-1}}$ to $v_{b_i}$ obtaining $K_1^{(i)}$, and compute the $m+1$ bundles consisting of each subwalk.
	More specifically, we compute $j_1^{(1)},\ldots,j_1^{(m)}$ and $j_2^{0},\ldots,j_2^{(m-1)}$ using Equation~\eqref{eqn_j1} and \eqref{eqn_j2}.

	The first part of the data structure stores the indices of the $m+1$ bundles using \DPT{}.
	These indices can be stored using
	\[
		\left\lceil 2\lg n^2+(m-1)\lg \left(\sum_x s_xt_x\right) \right\rceil+1
	\]
	bits of space, which by Equation~\eqref{eqn_num_bundle} is at most
	\begin{equation}\label{eqn_space1}
		4\lg n+(m-1)\lg\frac{\left<\sigma,\pi\right>\cdot n^4}{\left<\sigma,\bOne\right>\left<\pi,\bOne\right>}+O(m\cdot n^{-2})+2.
	\end{equation}

	\bigskip
	Next, we compute the index of the subwalk \emph{within} each bundle.
	More specifically, we compute $k_1^{(1)},\ldots,k_1^{(m)}$ and $k_2^{0},\ldots,k_2^{(m-1)}$ as follows:
	\[
		k_1^{(i)}=K_1^{(i)}-\left\lceil(j_1^{(i)}-1)\cdot \frac{\cN_{a_i-b_{i-1}}(v_{b_{i-1}},v_{a_i})}{s_x}\right\rceil,
	\]
	and
	\[
		k_2^{(i)}=K_2^{(i)}-\left\lceil(j_2^{(i)}-1)\cdot \frac{\cN_{b_i-a_{i}}(v_{a_{i}},v_{b_i})}{t_x}\right\rceil.
	\]
	They are the indices \emph{within} each $g_{x,j}$ and $h_{x,j}$.

	The second part of the data structure stores the subwalk between the milestones within the bundles.
	In particular, we store the triples $(v_{b_i},k_1^{(i)},k_2^{(i)})$ for every $i=0,\ldots,m-1$.
	By Equation~\eqref{eqn_gxy} and \eqref{eqn_hxy}, the number of triples is at most
	\begin{align*}
		&\sum_{y} \frac{\left<\sigma,\bOne\right>\pi_y}{\left<\sigma,\pi\right>}\cdot \lambda^{b_i-a_i}n^{-2}\cdot \frac{\sigma_y\left<\pi,\bOne\right>}{\left<\sigma,\pi\right>}\cdot \lambda^{a_{i+1}-b_i}\cdot n^{-2}\cdot (1\pm O(n^{-2})) \\
		&=\frac{\left<\sigma,\bOne\right>\left<\pi,\bOne\right>}{\left<\sigma,\pi\right>}\cdot \lambda^{a_{i+1}-a_i} n^{-4}(1\pm O(n^{-2})).
	\end{align*}


	Again, we store these triples using \DPT{}.
	The total space of the second part is at most
	\begin{align}\label{eqn_space2}
		m\lg \frac{\left<\sigma,\bOne\right>\left<\pi,\bOne\right>}{\left<\sigma,\pi\right>n^4}+n\lg\lambda+O(m\cdot n^2)+2.
	\end{align}

	Finally, summing up Equation~\eqref{eqn_space1} and \eqref{eqn_space2}, the total space usage of the data structure is at most
	\begin{align*}
		&\, 4\lg n+(m-1)\lg\frac{\left<\sigma,\pi\right>\cdot n^4}{\left<\sigma,\bOne\right>\left<\pi,\bOne\right>}+m\lg \frac{\left<\sigma,\bOne\right>\left<\pi,\bOne\right>}{\left<\sigma,\pi\right>n^4}+n\lg\lambda+O(m\cdot n^{-2})+4 \\
		\leq&\, \lg \frac{\left<\sigma,\bOne\right>\left<\pi,\bOne\right>}{\left<\sigma,\pi\right>}+n\lg\lambda+O(m\cdot n^{-2})+4 \\
		=&\, \lg\left( \frac{\left<\sigma,\bOne\right>\left<\pi,\bOne\right>}{\left<\sigma,\pi\right>}\cdot\lambda^n \right)+O(m\cdot n^{-2})+4 \\
		\intertext{which by Lemma~\ref{lem_power} again, is at most}
		\leq &\, \lg (\bOne^{\top}A^n\bOne)+5.
	\end{align*}

	\paragraph{Query algorithm.} Given an integer $q\in[0, n]$, we first compute the block that contains $v_t$, suppose $a_i\leq q\leq a_{i+1}$.
	Suppose $q$ is in the first half ($a_i\leq q\leq b_i$), we use the query algorithm of \DPT{} on the first part of the data structure, retrieving the index of the bundle containing the subwalk from $b_{i-1}$ to $b_i$, $(v_{a_i},j_1^{(i)},j_2^{(i)})$.
	Similarly, we retrieve the triple $(v_{b_i},k_1^{(i)},k_2^{(i)})$ from the second part.
	Then the encoding of the subwalk from $v_{a_i}$ to $v_{b_i}$ can be computed 
	\[
		K_2^{(i)}=k_2^{(i)}+\left\lceil(j_2^{(i)}-1)\cdot \frac{\cN_{b_i-a_{i}}(v_{a_{i}},v_{b_i})}{t_x}\right\rceil.
	\]
	Finally, we use the query algorithm of Lemma~\ref{lem_ram_decode} to decode the $(q-a_i)$-th vertex in this subwalk, which is $v_t$.
	The case where $q$ is in the second of the block ($b_i\leq q\leq a_{i+1}$) can be handled similarly.
	With an extra lookup table of size $r$ for Lemma~\ref{lem_ram_decode}, each query can be answered in $O(\frac{\lg\lg n}{\lg\lg r})$ time.
	This proves the theorem.
\end{proof}

The above data structure also extends to general directed graphs.

\begin{corollary}\label{cor_sc}
	With an extra $O(1)$ bits of space, Theorem~\ref{thm_general} also applies to general strongly connected graph $G$.
\end{corollary}
\begin{proof}
	Suppose $G$ is periodic with period $p$.
	Then its vertices can be divided into $p$ sets $V_1,\ldots,V_p$ such that each vertex $V_i$ only has outgoing edges to $V_{i+1}$ (defining $V_{p+1}=V_1$).
	We can reduce the problem to the aperiodic case as follows.

	Let us first only consider walks that both start and end with vertices in $V_1$.
	Therefore, the length of the walk $n$ must be a multiple of $p$.
	We divide the walk into subwalks of length $p$, and view each length-$p$ subwalk as a vertex in a new graph.
	That is, consider the graph $G^p|_{V_1}$ whose vertices are all possible length-$p$ walks that start and end with vertices in $V_1$ in $G$, such that directed edges connect two subwalks if they can be concatenated, i.e., if the last vertex in $u$ is the same as the first vertex in $u'$, then there is an edge from $u$ to $u'$.
	It is easy to verify that this graph is strongly connected and aperiodic, and a length-$n$ walk on $G$ can be viewed as a length-$n/p$ walk on this graph.
	Then we apply Theorem~\ref{thm_general} to store the length-$n/p$ walk on $G^p|_{V_1}$.
	To retrieve $v_q$, it suffices to retrieve the $\lfloor q/p\rfloor$-th vertex in the walk on $G^p|_{V_1}$, and recover $v_q$ via a look-up table.

	For general inputs that do not necessarily start or end in $V_1$, we store a prefix and suffix of walk of length at most $p$ naively, such that the remaining middle part starts and ends in $V_1$, for which we use the above construction.
\end{proof}

Finally, our data structure also applies to general directed graphs.
\begin{corollary}
	With an extra $O(\lg n)$ bits of space, Theorem~\ref{thm_general} also applies to general directed graph $G$.
\end{corollary}
\begin{proof}
	We precompute the strongly connected components (SCC) of $G$.
	Given input walk $(v_0,\ldots,v_n)$, observe that it can only switch between different SCCs a constant number of times, since after the walk leaves an SCC, it could never come back.
	Hence, we first store the steps in the input that switch from one SCC to another using $O(\lg n)$ bits, then apply Corollary~\ref{cor_sc} on each subwalk within an SCC.
\end{proof}

\subsection{Matching the point-wise optimal space for non-regular graphs}
The data structure from Theorem~\ref{thm_general} used approximately $\lg \bOne^\top A_G\bOne$ bits of space uniformly on \emph{any} possible input sequence. This means that our data structure achieves the best possible space for a \emph{uniformly chosen} 
length-$n$  walk in $G$, by Shannon's source coding theorem. 
However, the distribution we actually care about is that of a random walk of length $n$ starting from a random vertex in $G$, 
which for \emph{non-regular} graphs may have much lower Shannon entropy.
In the other words, for non-regular $G$, the following two distributions may be very far in KL divergence:   
\begin{itemize}
	\item $\nu_n$: A uniformly chosen length-$n$ walk in $G$.  
	\item $\mu_n$: An $n$-step random walk starting from a random vertex in $G$. 
\end{itemize}
By Huffman coding, $H(\mu_n)$ is therefore the correct \emph{expected} space benchmark. 
To achieve this expectation, the corresponding \emph{poin-wise} space benchmark per walk 
$(v_0,\ldots,v_n) \sim \mu_n$ in $G$ is 
\[
	\lceil\lg |G|+\sum_{i=0}^{n-1}\lg \deg(v_i)\rceil
\]
bits of space, since the space allocated to each sequence $(v_0,\ldots,v_n)$ is approximately $1/\lg \Pr[(v_0,\ldots,v_n)]$.

\bigskip

To match the above point-wise space bound, we use the \emph{augmented $B$-tree} from~\cite{Pat08}.
Fix a parameter $B\geq 2$, \cite{Pat08} defines data structure \emph{aB-trees} as follows:
\begin{itemize}
	\item The data structure stores an array $A\in \Sigma^n$. The data structure is a $B$-ary tree with $n$ leaves storing elements in $A$.
	\item Every node is augmented with a label from some alphabet $\Phi$, such that the label of the $i$-th leaf is a function of $A[i]$, and the label of an internal node is a function of the labels of its $B$ children, and the size of the subtree.
	\item The query algorithm examines the label of the root's children, decides which child to recurse on, examines all labels of that child's children, recurses to one of them, and so on.
	The algorithm must output the query answer when it reaches a leave.
\end{itemize}
\Pat proves a general theorem to compress \emph{any} such aB-tree, almost down to its input entropy:

\begin{theorem}[\cite{Pat08}, Theorem 8]\label{thm_abtree}
	Let $B\leq O(\frac{w}{\lg (n+|\Phi|)})$, and let $\cN(n, \varphi)$ be the number of instances of $A\in \Sigma^n$ that has the root labeled by $\varphi$. An aB-tree of size $n$ with root label $\varphi$ can be stored using $\lg_2 \cN(n, \varphi)+2$ bits. The query time is $O(\lg_B n)$, assuming a precomputed lookup tables of $O(|\Sigma|+|\Phi|^{B+1}+B\cdot |\Phi|^B)$ words, which only depend on $n$, $B$ and the aB-tree algorithm.
\end{theorem}

We build an aB-tree on top of the input. 
The alphabet size $|\Sigma|=|G|$.
For each node of the tree that corresponds to a subwalk $(v_l,\ldots,v_r)$, we set its label $\varphi$ to be the triple $(v_l, v_r, S)$, such that $S$ is an integer equal to
\[
	\sum_{i=l}^{r-1}\lceil n\lg_2 \deg(v_i)\rceil,
\]
\emph{with the exception of root}, whose label is only the integer $S=\sum_{i=0}^{n-1}\lceil n\lg_2 \deg(v_i)\rceil$, without $v_0$ and $v_n$.
In the other words, $S$ encodes the optimal ``space'' that this subsequence uses.
Thus, the alphabet size of the labels $|\Phi|$ is $O(n^2)$.
The label of a leaf $v_l$ is $(v_l, v_l, 0)$, a function of $v_l$.
To see why the label of an internal node is a function of the labels of its $B$ children, suppose the $B$ children have labels $(v_l,v_{m_1-1},S_1)$, $(v_{m_1},v_{m_2-1},S_2), \ldots$, $(v_{m_{B-1}},v_r,S_B)$ respectively.
Then the label of this node is $(v_l,v_r,S)$ (or just $S$ for the root) for 
\[
	S=\sum_{i=1}^B S_B+\sum_{i=1}^{B-1} \lceil n^2\lg_2\deg(v_{m_i-1})\rceil.
\]

Theorem~\ref{thm_abtree} compresses this data structure to $\lg_2 \cN(n+1, \varphi)+2$ bits, where $\phi$ is the label of the root, and supports queries in $O(\lg_B n)$ time.
The following lemma bounds the value of $\cN(n+1, \varphi)$.

\begin{lemma}\label{lem_space_abtree}
For root label $\varphi=S$, we have $\cN(n+1,\varphi) \leq 2^{\lg|G|+S/n}$.
\end{lemma}

\begin{proof}
Let $V$ be a uniformly random walk $V=(V_0,\ldots, V_n)$ with root label $\varphi$. 
Then $H(V) =  \lg_2\cN(n+1,\varphi).$ 
On the other hand, by the chain rule and the fact that conditioning only reduces entropy, we have 
\[ H(V) = H(V_0,\ldots, V_n) \leq H(V_0)+\sum_{i=0}^{n-1} H(V_{i+1}|V_i)  \leq \lg|G|+\sum_{i=0}^{n-1} \E[\lg \deg(v_i)] \leq \lg|G|+S/n, \] 
where the last inequality is by definition of $\cN(n,\varphi)$.   Rearranging sides completes the proof. 
\end{proof}

Therefore, the space usage is at most
\[
	\lg|G|+S/n+2\leq \lg|G|+\sum_{i=0}^{n-1}\lg\deg(v_i)+3.
\]

In particular, setting $B=2$ gives us query time $O(\lg n)$ and lookup table size $\tilde{O}(n^6)$.

\section{Lower Bounds} \label{sec_LB}

\subsection{A Succinct Reduction from Dictionaries to Directed Random Walks} \label{subsec_dictionary_vs_RW}

In this section, we exhibit a succinct reduction from the well-studied Dictionary data structure 
problem to the LDSC problem of storing random walks on (directed) non-regular graphs. 
Before we describe the reduction, we begin with several definitions. Recall that the zeroth-order 
\emph{empirical entropy} of a string $x\in \Sigma^n$ is $H_0(x) := \sum_{\sigma \in \Sigma} f_\sigma \lg(n/f_{\sigma}(x))$ 
where $f_\sigma$ is the number of occurrences (frequency) of the symbol $\sigma$ in $x$. 
The \emph{succinct Dictionary} problem is defined as follows: 

\begin{definition}[Succinct Dictionaries]\label{def_Dictionary}
Preprocess an $n$-letter string $x\in \Sigma^n$ into $H_0(x) + r$ bits of space, 
such that each $x_i$ can be retrieved in time $t$. We denote this problem $\Dic$. 
\end{definition} 

Note that in the special case of bit-strings ($\Sigma = \{0,1\}$), this is the classic 
\emph{Membership} problem \cite{Pagh02, gal:succinct, Pat08}.  
The best known time-space tradeoff for $\Dic$ (for constant-size alphabets)  
is $r = O(n/(\frac{\lg n}{t})^t)$ \cite{Pat08}.\footnote{On a RAM, there is an extra $n^{1-\eps}$ additive 
term for storing lookup-tables whose ``amortized" cost is small, see Theorem 1 in \cite{Pat08}.} 
While this exponential tradeoff is known to be optimal in the bit-probe model ($w=1$) \cite{Vio09}, 
no cell-probe lower bounds are known, and this is one of the long-standing open problems 
in  the field of succinct data structures. \\

A key component in our reduction will be the \emph{Huffman tree} (c.f, Huffman code) of a distribution.

\begin{definition}[Huffman Tree]\label{def_huff_tree}
The \emph{Huffman Tree} $\cT_\mu$ of a discrete distribution $\mu$ supported on alphabet $\Sigma$, 
is a rooted binary tree (not necessarily complete) with $\Sigma$ leaves, one per alphabet symbol.  
For every symbol $\sigma \in \Sigma$, the unique root-to-leaf path to the symbol $\sigma$ is of 
length exactly $\ell_\sigma := \lceil \lg_2(1/\mu(\sigma)) \rceil$. 
\end{definition} 

We remark that the existence of such trees for any discrete distribution $\mu$ is guaranteed by Kraft's 
inequality (for more context and construction of the Huffman tree, see \cite{TC06}). 
We are now ready to prove Theorem \ref{thm_LB_informal}, which we restate below. 

\begin{theorem}\label{thm_Dic_to_nonreg_RW} 
Let  $D$ be a succinct cell-probe data structure for storing a walk $(v_1,\ldots,v_n)$ over general (directed) 
graphs $G$, using $\sum_i \lg(deg(v_i)) + r$ bits of space, and query time $t$ for each $v_i$. 
Then for any constant-size alphabet $\Sigma$, there is a succinct dictionary storing $x\in \Sigma^n$, 
with space $H_0(x)  + r$ bits and query time query time $t +1$. This reduction holds for any input $x$ with 
empirical frequencies ($f_{\sigma}(x)/n$) which are inverses of powers of $2$ lower bounded by a constant.  
\end{theorem} 

\begin{proof}
First, we claim that the Dictionary problem can be reformulated as follows: Store an $n$-letter string 
drawn from some arbitrary \emph{product} distribution $X \sim \mu^n$,  
using expected space $s=nH(\mu) + r$ bits, decoding $X_i$ in time $t$, where $H(\mu)$ is the Shannon entropy of $\mu$. 
Indeed, let $\mu=\mu(x)$ denote the empirical distribution of the input string $x\in \Sigma^n$ to $\Dic$. 
Since the zeroth-order space benchmark $H_0(x)$ only depends on the \emph{marginal} frequencies of symbols in the 
input string, 
we may assume that each coordinate $X_i$ is drawn \emph{independently} from 
the induced distribution $\mu$ (which is arbitrary nonuniform), and $H_0(x) = H(\mu^n) = nH(\mu)$ holds. 

By this reformulation, we may assume the input to $\Dic$ is $X\sim \mu^n$ for some $\mu$. We now define a (directed) graph 
such that a uniform random walk on this graph encodes the input string $X_1,\ldots,X_n \sim \mu^n$  
\emph{without losing entropy}. Let $\cT_\mu$ be the Huffman Tree of the distribution $\mu$, where all 
edges are directed downwards from root to leaves. By Definition \ref{def_huff_tree}, the probability 
that a random walk $W$ on $\cT_\mu$ starting from the root reaches leaf $\sigma\in \Sigma$ is precisely 
$2^{-\ell_\sigma} = 2^{-\lceil \lg_2(1/\mu(\sigma)) \rceil}$. If all probabilities in $\mu$ are (inverses of) powers 
of $2$, then the probability of this event is exactly 
\begin{align}\label{eq_t_mu_prob}
\Pr[W \text{reaches leaf } \sigma] =  2^{-\lg_2(1/\mu(\sigma))} = \mu(\sigma). 
\end{align}
In other words, under the assumption that all probabilities are powers of $2$, sampling $X_i\sim\mu$ is equivalent 
to sampling a random leaf of $\cT_\mu$ according to a uniform random walk starting from the root $v_r$. 
Let $$d_\mu := \max_\sigma \ell_\sigma$$ be the height of the tree $\cT_\mu$ (i.e., the length of the deepest path 
$\max_\sigma \lg(1/\mu(\sigma))$). To complete the construction of the graph, we connect each leaf (corresponding to symbol) 
$\sigma$ back to the root of $\cT_\mu$ via a vertex-disjoint \emph{directed path} of length $d_\mu + 1 - \ell_\sigma$.  
Call the resulting graph $G_\mu$. 
This simple trick makes sure that all ``cycles" have \emph{fixed} length $d_\mu+1$, hence each sampled leaf 
in the walk can be decoded from a fixed block, despite the fact that each leaf is sampled at unpredictable depth.
The key observation is that adding these vertex-disjoint paths does \emph{not} increase the entropy of a 
random walk ($\lg_2(1) = 0$), but merely increases the size of the graph by a constant factor 
$|G_\mu| = O(|\Sigma| d_\mu) = O(1)$.

Let $(V_1,\ldots, V_{n'})$ be a uniform random walk of length $n':= (d_\mu +1)n$ on $G_\mu$ starting 
from the root. By \eqref{eq_t_mu_prob} and definition of $G_\mu$, it follows that 
the unique leaf $V_{j_i} \in \Sigma$  
sampled in the $i$th ``cycle" of the walk ($j_i \in [(i-1)(d_\mu+1), i (d_\mu+1)]$ ), is distributed precisely as $X_i\sim\mu$.  
In particular, $H(V_{j_1},\ldots, V_{j_n}) = nH(\mu)$, while the (conditional) entropies of all the rest of the 
$V_j$'s are identically $0$ (as all previous vertices $V_{<j_i}$ in the $i$th ``cycle" are determined by $V_{j_i}$ 
by the tree structure, and all remaining vertices $V_{>j_i}$ in the $i$th cycle
correspond to the part of the walk that is a vertex-disjoint \emph{path} and hence the conditional entropy 
of this subwalk is $0$).   
Thus, by the chain rule and the premise of the theorem,  
we conclude that there is a succinct data structure storing a random walk $V_1,\ldots, V_n$ on $G_\mu$
using space $nH(\mu) + r$ bits of space, which by the above reformulation implies the same space 
for storing $x \in \Sigma^n$ up to the $O(|\Sigma|\lg n)$ additive term,  
so long as marginal frequencies ($f_{\sigma}(x)/n$) are inverses of powers of $2$.  

Decoding of $x_i$ is straightforward: Go to the second-before-last vertex $V_{i(d_\mu+1) -1}$ in the $i$th 
cycle (i.e., the one just before coming back to the root); Either $V_{i(d_\mu+1) -1}$ is a leaf of $\cT_\mu$, 
or it belongs to a \emph{unique} path corresponding to some leaf $\ell_\sigma$ (as all back-tracking 
paths are vertex disjoint) hence we know immediately that $X_i = \sigma$.  By the premise of  the theorem, 
the decoding time of each vertex in the walk is $t$, hence so is that of $X_i$. We conclude that the resulting 
succinct data structure is an $(r,t)$-Dictionary, as claimed.


\end{proof}


\subsection{Completeness of $\LDSC$ in the static cell-probe model}  \label{subsec_completeness_LDSC}
In this section, we prove unconditional (cell-probe) lower bounds on the $\LDSC$ problem, 
demonstrating that for some (in fact, most) classes of non-product distributions, this storage problem does 
not admit efficient time-space tradeoffs. 
Our first observation is that the $\LDSC$ problem under an \emph{arbitrary} joint distribution 
$\mu$ is a ``complete" static data structure problem: 

\begin{proposition}[Completeness of $\LDSC$]  \label{prop_completeness_ldsc}
For any prior $\mu$ on $n$ files, $\LDSC^{\mu}_{n,\Sigma}$ is equivalent to some 
static data structure problem $\cP = \cP(\mu)$ with $|\cQ| = n$ queries, on an input 
database $Y$  
of size $\E[|Y|] \leq H_\mu + 1$ bits. Conversely, any static data structure 
problem $\cP$ with $|\cQ|$ queries on a database $x \in \Sigma^n$, 
can be embedded as an $\LDSC$ problem on some $\mu = \mu(\cP)$, 
where $H_\mu(\X) = n$. 
These equivalences hold in the cell-probe model with word size $w \geq \Omega(\lg |\Sigma|)$.
\end{proposition}

\begin{proof}
The converse statement is straightforward: Given any static problem $\cP$
on input $x \in \Sigma^N$ with query set $\cQ$, consider the distribution $\mu(\cP)$ 
of $(q_1(X),q_2(X),\ldots , q_{|\cQ|}(X))$, i.e., of all query answers $(A_1,\ldots ,A_{|\cQ|})$ 
to problem $\cP$ under a uniformly random input database $X \in_R \Sigma^n$. The joint entropy 
of $\mu(\cP)$ is $H(A_1,\ldots ,A_{|\cQ|}) \leq H(X) = N\lg \Sigma$, since 
$X$ determines all query answers. Hence this is a valid LDSC instance on $n = |\cQ|$ 
files, over a joint distribution $\mu(\cP)$ with entropy $H_\mu \leq N\lg \Sigma \leq n$. 
    
For the first direction of the proposition, let $\X := (X_1,\ldots,X_n)$ be the (random variable)  
input files to $\LDSC^{\mu}_{n,\Sigma}$, and let $Y := \Huff(X_1,\ldots,X_n)$ denote the 
\emph{Huffman code} of the random variable $\X$. By the properties of the Huffman code (c.f. Definition \ref{def_huff_tree}), 
$Y = \Huff(\X)$ is \emph{invertible} and has expected length $\E_{\mu}[Y] \leq H_\mu(\X) + 1$ bits.  
This implies that $Y$ is a uniformly distributed random variable with entropy $H(Y)=H(\X)$. 
Now, define the data structure problem $\cP$, in which the answer to the $i$th query is
$\cP(i,Y) := (\Huff^{-1}(Y))_i = X_i$ where the last inequality follows from the fact that 
Huffman coding is lossless and hence $\X$ is determined by \emph{some} deterministic 
function $g(Y)$. While this provides a bound only on the expected size of the input (``database") 
$Y$, a standard Markov argument can be applied for bounding the worst-case size, if one 
is willing to tolerate an arbitrarily small failure probability of the data structure over the distribution $\mu$.

\end{proof}

\begin{corollary}[LDSC Lower Bounds] \label{cor_LDSC_LBs}
Proposition \ref{prop_completeness_ldsc}  has the following lower bound implications: 
\begin{enumerate}   
\item ($k$-wise independent distributions)      
Using the cell-sampling arguments of \cite{Siegel04, Larsen12a},   
the above reduction implies that any linear-space data-structure (storing $O(H_\mu)$ bits) 
for \probarg{\LDSC^{\mu}}{n,\Sigma} 
when \emph{$\mu$ is an $(H_\mu)$-wise independent} distribution, 
requires $$t\geq \Omega(\lg H_\mu)= \Omega(\lg n)$$ decoding time even in the cell-probe model (note that this type of tradeoff 
is the highest explicit lower bound known for any static problem). 
\item A rather simple counting argument (\cite{Milt93}) implies that for \emph{most} joint distributions  $\mu$ with 
entropy $H(\mu) := h \ll n$, locally decodable source coding is impossible, in the sense that decoding requires 
$h^{1-o(1)}$ time unless trivial $\approx n^{1-o(1)}$ space is used (and this is tight by Huffman coding). 
Such implicit hard distribution $\mu$ can be defined by $n$ random functions on a random $h$-bit string, where $h = n^\eps$. 
\end{enumerate}
\end{corollary}



\bibliography{walk_refs.bib}
\bibliographystyle{alpha}

\end{document}